\documentclass[,draftclsnofoot, 12pt, english, onecolumn]{IEEEtran}
\usepackage{cite}
\usepackage[T1]{fontenc}
\usepackage[latin9]{inputenc}
\usepackage{amsmath}
\usepackage{amssymb}

\makeatletter
\usepackage{amsthm}\usepackage{dsfont}\usepackage{array}\usepackage{mathrsfs}\usepackage{graphicx}
\makeatother

\usepackage{babel}

\begin{document}
\bibliographystyle{IEEEtran}

\title{On the Accuracy of the Wyner Model in Cellular Networks}
\author{Jiaming Xu, Jun Zhang and Jeffrey G. Andrews  \\
\thanks{ J. Xu and J. G. Andrews are with The University of Texas at Austin. Email: jxu@mail.utexas.edu and jandrews@ece.utexas.edu. J. Zhang is with The Hong Kong University of Science and Technology. Email: eejzhang@ust.hk. Contact author is J. G. Andrews. Date revised: \today}
}
\maketitle
\newtheorem{lemma}{Lemma}\newtheorem{theorem}{Theorem}\newtheorem{corollary}{Corollary}\newtheorem{remark}{Remark}

\begin{abstract}
The Wyner model has been widely used to model and analyze cellular networks due to its simplicity and analytical tractability.  Its key aspects include fixed user locations and the deterministic and homogeneous interference intensity.  While clearly a significant simplification of a real cellular system, which has random user locations and interference levels that vary by several orders of magnitude over a cell, a common presumption by theorists is that the Wyner model nevertheless captures the essential aspects of cellular interactions.  But is this true? To answer this question, we consider both uplink and downlink transmissions, and both outage-based and average-based metrics.  For the uplink, for both metrics, we conclude that the Wyner model is in fact quite accurate for systems with a sufficient number of simultaneous users, e.g. CDMA. Conversely, it is broadly inaccurate otherwise. With multicell processing, intracell TDMA is shown to be suboptimal in terms of average throughput, in sharp contrast to predictions using the Wyner model. Turning to the downlink, the Wyner model is highly inaccurate for outage since it depends largely on the user locations. However, for average or sum throughput, the Wyner model serves as an acceptable simplification in certain special cases if the interference parameter is set appropriately.
\end{abstract}

\section{Introduction}

Intercell interference is a key bottleneck of modern cellular networks. As a consequence, multicell processing (MCP), which can efficiently suppress intercell interference, is of great current interest to both researchers and cellular standard bodies \cite{Shamai04,Zhang04,Foschini06,Liang06,Goldsmith06,Choi07,Vrzic09,Zhang09,Ghosh10}. Surprisingly, analytical techniques for modeling such interference are in short supply. Grid-based models of multicell networks are too complex to handle analytically so researchers, particularly in the pursuit of information theoretic results, often resort to a simple multicellular model originally proposed by Wyner \cite{Wyner94}. Compared to real cellular networks, the most popular linear version of the Wyner model (see Fig. \ref{FigWynerModel}) makes three major simplifications: (i) only interference from two adjacent cells is considered; (ii) random user locations and therefore path loss variations are ignored; (iii) the interference intensity from each neighboring base station (BS) is characterized by a single fixed parameter $\alpha$ ($0 \le \alpha \le 1$).

\subsection{Backgound and Related Work}
This model was first used in \cite{Wyner94} to derive the capacity of uplink cellular networks with multicell processing (MCP), where it is shown that intracell TDMA is optimal and achieves the capacity. It was generalized in \cite{Shamai00} to account for flat-fading. It is proved that wideband transmission is advantageous over intracell TDMA and that fading increases capacity when the number of users is sufficiently large. In \cite{Shamai97} and \cite{Shamai97Part2}, the Wyner model is used to analyze the throughput of cellular networks under single-cell processing (SCP) and two-cell-site processing (TCSP). Later on, scaling results for the sum capacity were derived under the Wyner model with multiple-input-multiple-output (MIMO) links in \cite{Hanly06}. Recently, the Wyner model is extended to incorporate shadowing in \cite{Kaltakis09}.

Turning to the downlink, the Wyner model was used to analyze the average throughput with MCP in \cite{Shamai01}, where it is shown that a linear preprocessing and encoding scheme eliminates the intercell interference and dramatically improves downlink cellular performance compared to the conventional SCP approach. Thereafter, it was generalized to derive the downlink sum capacity \cite{Shamai07} and sum throughput of various precoding schemes such as zero-forcing (ZF), beamforming (BF), minimum mean square error (MMSE) \cite{Padovani08} and ZFBF \cite{Somekh06,SimeoneIT09July}. More recently, the Wyner model and its modifications have been widely used to evaluate different constrained coordination strategies in cellular networks, e.g., \cite{Simeone09,Sanderovich09,SimeoneIT09,Levy10}.

Despite the fairly large amount of literature based on the Wyner model, to our knowledge no serious effort has ever been made to validate this simple model, or to understand when it might be a reasonable approximation. Namely, how much is lost from these three simplifications? What might be an appropriate value of the key parameter $\alpha$? It is not at all clear when conclusions based on the Wyner model broadly hold for more realistic scenarios. In this paper, we attempt to provide detailed analysis and simulations to determine when the Wyner model is accurate, and when it is not. ``Accurate'' is a necessarily subjective term here which we use to mean that major trends are captured within a ``small'' constant, for example a metric can be tracked within a factor of some small constant.  ``Inaccurate'' means that key trends are in our opinion totally missed and the error cannot in general be bounded within a small constant, and results may be off by an order of magnitude or more.  Different audiences may care to draw their own conclusions on its relative accuracy: this paper intends to offer a mathematics-based starting point for such a judgment, and to provide initial conclusions as to when it \emph{may} be sufficiently accurate and also when it is almost \emph{certainly inaccurate}.

\subsection{Contributions}
We compare a $1$-D Wyner model to a $1$-D grid-based network with random user locations. In the first part of this paper, we study the uplink channel. Two important performance metrics in cellular systems are outage probability and (average or sum) throughput. From an outage and average throughput point of view, we show that with single-cell processing, the Wyner model is broadly inaccurate for systems with a small number of simultaneous users because the statistics cannot be captured by the single parameter $\alpha$. On the other hand, if there are many simultaneous users, the Wyner parameter $\alpha$ can be accurately tuned to capture key trends in outage and throughput. In cellular parlance, this means that TDMA and other orthogonal multiple access protocols evade characterization by the Wyner model, whereas CDMA and/or other non-orthogonal uplink protocols can be captured by it. In the case of multicell processing, we show that CDMA outperforms intracell TDMA, which is in direct contradiction to the original result of Wyner, where intracell TDMA is shown to be optimal and capacity-achieving \cite{Wyner94}.

In the second part of this paper, we turn to the downlink channel,  and find conclusions quite different from the uplink case. For outage, we show that with single-cell processing and perfect channel inversion, the Wyner model fails to capture significant variations over a cell in both TDMA and CDMA multicell networks. For throughput, with channel inversion and by tuning the parameter $\alpha$ appropriately, the Wyner model can characterize average throughput in CDMA with reasonable accuracy, but it is still inaccurate for those employing TDMA.  Furthermore, if equal transmit power per user is assumed instead, the Wyner model becomes much less accurate for both TDMA and CDMA. In the case of multicell processing with equal power allocation, it is shown that the Wyner model is quite accurate for sum throughput in CDMA with a properly selected $\alpha$, which we derive.

The unifying conclusion is that for certain scenarios, the Wyner model, which itself models an ``average'' interference condition, can be adapted to handle metrics like sum or average throughput which average over the cell.  But it cannot handle metrics that depend on user locations like outage probability with the exception of CDMA uplink in which case the Wyner model is accurate.

The rest of this paper is organized as follows. Section \ref{Model} introduces the system model and key definitions. Section \ref{UplinkSCP} and Section \ref{UplinkMCP} are devoted to the analysis of the uplink channel with single-cell processing and multicell processing respectively, while Section \ref{DownlinkSCP} and Section \ref{DownlinkMCP} present the corresponding analysis of the downlink channel. Section \ref{OFDMA} extends the results to OFDMA systems. Section \ref{conclusion} ends the paper with a summary and concluding remarks. Proofs are deferred to the appendices.

\section{System Model and Definition} \label{Model}
 The system model and key definitions in this paper are described in this section.
\subsection{The System Model}
Consider the $1$-D linear Wyner model, depicted in Fig.~\ref{FigWynerModel}, where there are $N$ cells located on a line, indexed by $ n \; (1\leq n\leq N)$, each covering a segment of length $2R$. BSs are located at the center of each cell and there are $K$ active users per cell. Intercell interference only comes from two neighboring cells. The channel gain between BS and its home user is $1$ and the intercell interference intensity is characterized by a deterministic and homogeneous parameter $\alpha$. Therefore, the locations of users are implicitly fixed.

In this paper, random user locations are considered and the network is redrawn in Fig.~\ref{FigWynerModelRandom}. In contrast to fixed user locations and deterministic and homogeneous intercell interference intensity, we assume all the users are randomly, independently, and uniformly distributed on the line in each cell. Thus, the intercell interference intensity is a random variable depending on user locations. The channel strength in this paper is assumed to be solely determined by pathloss, i.e., the received power at distance $r$ decays according to $r^{-\beta}$, where $ \beta \ge 2$ is the pathloss exponent. Fading is not considered, since it can be handled with standard diversity counter measures. For simplicity, lognormal shadowing is also neglected.

Under these assumptions, for the uplink, the received signal at BS $n$ at a given time is given by
\begin{align}
Y_{n}  = \underbrace{ \sum_{k=1}^{K} \sqrt {P_{n}^{(k)} (r_{n,n}^{(k)})^{-\beta}} X_{n}^{(k)}}_{\text{home user signals}}
+ \underbrace{\sum_{i=1}^{K} \sqrt {P_{n-1}^{(i)} (r_{n-1,n}^{(i)})^{-\beta} }  X_{n-1}^{(i)}}_{\text{interference from cell} \; n-1}
  + \underbrace{ \sum_{j=1}^{K} \sqrt {P_{n+1}^{(j)} (r_{n+1,n}^{(j)})^{-\beta} } X_{n+1}^{(j)} }_{\text{interference from cell} \; n+1}+ Z_{n}, \nonumber
\end{align}
\noindent where $Y_{n}$ is the received signal at BS $n$, $X_{n}^{(k)}$ is the transmit signal from user $k$ in cell $n$ with $\mathbb{E} [|X_{n}^{(k)}|^{2}] =1 $, $  \{Z_{n} \}$ is mutually independent zero-mean additive white Gaussian noise with variance $\sigma^2$, $r_{n,m}^{(k)}$ is the distance from user $k$ in cell $n$ to the BS $m$, and $P_{n}^{(k)} $ is the transmit power of user $k$ in cell $n$. Every user is assumed to achieve perfect channel inversion in the uplink, i.e., the received signal power of each user at its home BS is $P$:
\begin{align}
 P_{n}^{(k)} (r_{n,n}^{(k)})^{-\beta} = P.  \label{EquationReceivedPower}
\end{align}

For the downlink, the received signal of user $k$ in cell $n$ at a given time is given by
\begin{align}
Y_{n}^{(k)}  = a_n^{(k)}  X_{n}
+    b_n^{(k)} X_{n-1}
 +   c_n^{(k)} X_{n+1}  + Z_{n}^{(k)}, \label{EquationWynerModelDownlink}
\end{align}
\noindent where $Y_{n}^{(k)} $ is the received signal at user $k$ in cell $n$, $X_{n}$ is the transmit signal of BS $n$. $  \{ Z_{n}^{(k)} \}$ is mutually independent zero-mean additive white Gaussian noise with variance $\sigma^2$. Channel gains from BS $n$, $n-1$, $n+1$ to user $k$ in cell $n$ are denoted by  $ \{ a_n^{(k)},b_n^{(k)},c_n^{(k)} \}$ respectively, which are given by
\begin{align}
a_n^{(k)}= \left( r_0/ r_{n,n}^{(k)} \right)^{ \frac{\beta} {2} }, \quad b_n^{(k)}=\left( r_0/ r_{n,n-1}^{(k)}  \right)^{\frac{\beta} {2} },  \quad c_n^{(k)}= \left( r_0 /r_{n,n+1}^{(k)}   \right)^{\frac{\beta} {2}}, \nonumber
\end{align}
where $r_0$ is the received power reference distance.

Let $L_{n}^{(k)}$ denote the location of user $k$ in cell $n$. Since users are assumed to be independently and uniformly distributed over cells, $ \{L_{n}^{(k)} \}$ are i.i.d. random variables over $n$ and $k$ with distribution $ \mathbb{U}\:[-R, R]$, where $\mathbb{U}$ denotes the uniform distribution. As depicted in Fig.~\ref{FigWynerModelRandom}, $r_{n,n}^{(k)}, r_{n,n-1}^{(k)},r_{n,n+1}^{(k)}$ can be written as a function of $L_{n}^{(k)}$ and $R$, which are
\begin{align}
r_{n,n}^{(k)}= | L_{n}^{(k)} |, \quad r_{n,n-1}^{(k)}= 2R+ L_{n}^{(k)}, \quad r_{n,n+1}^{(k)}= 2R- L_{n}^{(k)}. \nonumber
\end{align}
\noindent Also, it follows that $ \{ r_{n,m}^{(k)} \}$ are independent over $n$ and $k$.

\subsection{Terminology}
In this paper, we consider various different system settings, which are explained as follows.
\begin{enumerate}
\item {\bf Intracell TDMA}: one user per cell is allowed to transmit at any time instant while users in different cells can transmit simultaneously.
\item {\bf CDMA}: all users in every cell are allowed to transmit simultaneously over the whole bandwidth. For the uplink, asynchronous CDMA is often used with pseudonoise sequences as signatures, so there exists intra-cell interference; while for the downlink, synchronous CDMA is often employed with orthogonal spreading codes and synchronous reception, so there is no intra-cell interference.
\item {\bf OFDMA}: OFDM is employed at the physical layer and all users in every cell are allowed to transmit simultaneously but every user in each cell is assigned a distinct subset of subcarriers, which eliminates the intra-cell interference.
\item {\bf Single-cell Processing (SCP)}: for the uplink, BSs only process transmit signals from their own cells and treat intercell interference as Gaussian noise; while for the downlink, BSs transmit signals with information only intended for their home users.
\item {\bf Multicell Processing (MCP)}: for the uplink, a joint receiver has access to all the received signals and an optimal decoder decodes all the transmit signals jointly; while for the downlink, the transmit signal from each BS contains information for all users.
\end{enumerate}

The relevant performance metrics used in this paper are as follows.
\begin{enumerate}
\item { \bf Outage Probability (OP) }: the probability that the received signal-interference-ratio (SIR) is smaller than a threshold.
\item {\bf Average Throughput (AvgTh)}: the expected user throughput where the expectation is taken over all possible user locations.
\item {\bf Sum Throughput (SumTh)}: the maximum sum rate of all users.
\end{enumerate}

\section{Uplink with SCP} \label{UplinkSCP}
In this section, we will study the accuracy of the Wyner model for describing the uplink channel in cellular networks with SCP. Two common multiaccess schemes (intracell TDMA and CDMA) are considered in this section and two key performance metrics, namely outage probability and average throughput, are investigated separately.

\subsection{Outage Probability} \label{Section_Uplink_SCP_Outage}

Consider intracell TDMA first. Since the interference-limited case is of interest in SCP, Gaussian noise $Z_{n}$ is ignored in this section. It follows that the signal-interference-ratio (SIR) of user $k$ in cell $n$ is a random variable given by
\begin{align}
\text{SIR}_{n}^{(k)}  = \frac{ P_n^{(k)} (r_{n,n}^{(k)})^{-\beta} }{ P_{n-1}^{(i)} (r_{n-1,n}^{(i)})^{-\beta}+ P_{n+1}^{(j)} (r_{n+1,n}^{(j)})^{-\beta} }. \nonumber
\end{align}
Substituting (\ref{EquationReceivedPower}), the SIR can be simplified to
\begin{align}
\text{SIR}_{n}^{(k)} = \left[  U_{n-1}^{(i)} + V_{n+1}^{(j)} \right]^{-1}, \label{Equation:TDMASIR}
\end{align}
\noindent where $U_{n-1}^{(i)}=\left(\frac{r_{n-1,n-1}^{(i)} } {r_{n-1,n}^{(i)} } \right)^{\beta}$ and $V_{n+1}^{(j)}=\left( \frac{r_{n+1,n+1}^{(j)} } {r_{n+1,n}^{(j)} } \right)^{\beta}$, and note that $\{U_{n}^{(i)}\}$ and $ \{V_{n}^{(j)}\}$ are i.i.d. random variables over $n$ and $i,j$ respectively. Also, $U_{n}^{(i)}$ and $V_{m}^{(j)}$ are i.i.d. for any $i$, $j$, and $n \ne m$.

To decode the received signal reliably, the SIR is constrained to be larger than a threshold $\theta$. If not, there arises outage and the outage probability $q$ is given by
\begin{align}
q & =  \mathbb{P}[\text{SIR}_{n}^{(k)}<\theta] = \mathbb{P} \left[U_{n-1}^{(i)}+V_{n+1}^{(j)} > \frac{1}{\theta} \right]. \nonumber
\end{align}

\noindent The outage probability can be lower bounded as
\begin{align}
 q  \ge \mathbb{P} \left[U_{n-1}^{(i)} > \frac{1}{\theta} \cup V_{n+1}^{(j)} > \frac{1}{\theta} \right]
                 =1-\mathbb{P}\left[U_{n-1}^{(i)} \le \frac{1}{\theta} \cap V_{n+1}^{(j)} \le \frac{1}{\theta} \right]
                = 1-\mathbb{P}^2 \left[U_{n-1}^{(i)} \le \frac{1}{\theta} \right],  \label{Probout}
\end{align}
where the last equality follows from the fact that $U_{n}^{(i)}$ and $V_{n+1}^{(j)}$ are i.i.d.. We also have
\begin{align}
\mathbb{P} \left [ U_{n-1}^{(i)} \le \frac{1}{\theta} \right ]
            = \mathbb{P} \left[ \frac{r_{n-1,n-1}^{(i)} } {r_{n-1,n}^{(i)} } <  \left(\frac{1}{\theta}\right)^{\frac{1}{\beta}} \right]
             = \mathbb{P} \left[ \frac{|L_{n-1}^{(i)}|}{d-L_{n-1}^{(i)}} <  \left(\frac{1}{\theta}\right)^{\frac{1}{\beta}} \right]. \nonumber
\end{align}

\noindent Since $L_{n}^{(i)} \sim  \mathbb{U}\:[-R, R]$, we can derive
\begin{align}
\mathbb{P} \left [ U_{n-1}^{(i)} \le \frac{1}{\theta} \right ] = \left \{
\begin{array}{ll}
 \frac{\theta^{- \frac{1}{\beta}}}{1+\theta^{- \frac{1}{\beta}}} + \frac{\theta^{- \frac{1}{\beta}}}{1 - \theta^{- \frac{1}{\beta}}} &   \theta > 3^\beta \\
 \frac{\theta^{- \frac{1}{\beta}}}{1+\theta^{- \frac{1}{\beta}}}+ \frac{1}{2}  &  1< \theta \le 3^\beta  \\
 1  & \theta \le 1 \\
\end{array} \right. , \nonumber
\end{align}
and the outage probability $q$ is lower bounded by (\ref{Probout}). The lower bound and simulation results of the outage probability are shown in Fig.~\ref{FigTDMA}. It can be seen that the lower bound is tight especially for a large pathloss exponent.

In the $1$-D Wyner model, $\text{SIR} \equiv \frac{1}{2 \alpha^2}$ and $\alpha$ is a deterministic parameter. Therefore, there is no notion of outage. However, when considering random user locations, intercell interference and thus SIR become random. The outage probability changes with the SIR threshold as shown in Fig.~\ref{FigTDMA}. Therefore, the Wyner model is inaccurate to characterize the outage probability in the intracell TDMA system.

In the following, let us consider the asynchronous CDMA where BSs will receive intercell interference from all user transmissions in two neighboring cells. Therefore, following the derivations in the intracell TDMA, the SIR at BS $n$ in CDMA can be derived as
\begin{align}
\text{SIR}_n =\frac{ G }{(K-1) + \sum_{k=1}^{K} (U_{n-1}^{(k)}  + V_{n+1}^{(k)})   },   \nonumber
\end{align}
where $G$ is the processing gain in CDMA. When $K \gg 1 $, since $U_{n-1}^{(k)}$ and $V_{n+1}^{(k)}$ are i.i.d., by the central limit theorem, $ \frac{1}{\sqrt{K} }  \sum_{k=1}^{K} ( U_{n-1}^{(k)}  + V_{n+1}^{(k)}) $ can be approximated by a Gaussian random variable with mean $\mu$ and variance $ \sigma_u^2$, which can be easily derived as
\begin{align} \label{EquationSigmaU}
\mu &=  2\sqrt{K}  \mathbb{E} [U_{n}^{(k)} ] \approx 2 \sqrt{K} \frac{\mathbb{E} [ |L_{n}^{(k)}| ] } {d^{\beta}} = 2 \sqrt{K} \frac{1}{(\beta+1) 2^{\beta}} , \nonumber \\
\sigma_u^2 & = 2 \left( \mathbb{E} [(U_{n}^{(k)})^2]- \mathbb{E}^2[U_{n}^{(k)} ] \right)
         \approx \frac{2}{2^{2 \beta}} \left[ \frac{1}{2 \beta +1}- \frac{1}{(\beta+1)^2} \right] ,
\end{align}
where the approximation is explained in Appendix \ref{ApproximationAlpha}.
\noindent Therefore, the outage probability can be derived as
\begin{align}
q   = \mathbb{P} \left[ \sum_{k=1}^{K} ( U_{n-1}^{(k)}  + V_{n+1}^{(k)})  > \frac{G}{\theta} - (K-1) \right ] \approx \mathcal{Q} \left(\frac{ G/{\theta} -(K-1) - \sqrt{K} \mu}{\sqrt{K} \sigma_u} \right), \label{OutageProb}
\end{align}
\noindent where $\mathcal{Q}(x)$ is the complementary cumulative distribution function (CCDF) of the standard Gaussian random variable.

Define the average intercell interference intensity with random user locations $\bar{\alpha}^{2}_{\text{ul}}$ as
\begin{align} \label{EquationDefAlphaUplinik}
\bar{\alpha}^{2}_{\text{ul}}= \mathbb{E} [U_{n}^{(k)} ]=\mathbb{E} [V_{n}^{(k)} ] \approx\frac{1}{(\beta+1)2^{\beta}}.
\end{align}
When $K \to \infty$, due to the law of large numbers:
\begin{align}
\lim_{K \to \infty} \frac{1}{K} \sum_{i=1}^{K} U_{n}^{(i)}  = \lim_{K \to \infty} \frac{1}{K} \sum_{j=1}^{K} V_{n+1}^{(j)}=  \bar{\alpha}^{2}_{\text{ul}}. \nonumber
\end{align}
Then SIR becomes deterministic and homogeneous, which is given by
\begin{align}
\text{SIR}_n=\frac{G}{K-1+ 2 \bar{\alpha}^{2}_{\text{ul}} K} \label{EquationSIRTDMA}.
\end{align}
Thus, the outage probability becomes either $0$ or $1$.

The equation (\ref{EquationSIRTDMA}) is the same as that derived under the Wyner model, if $\alpha = \bar{\alpha}_{\text{ul}}$. In other words, when $K$ is sufficiently large, the Wyner model is accurate and the parameter $\alpha$ can be tuned to characterize intercell interference intensity quite well. If users are uniformly distributed in each cell, $\alpha$ can be approximately determined by $\alpha^{2}= \mathbb{E} [U_{n}^{(k)} ] \approx\frac{1}{(\beta+1)2^{\beta}}$ . From this expression, we can see that when the path loss exponent $\beta$ increases, $\alpha^2$ decreases exponentially fast.

Simulation and numerical results according to equation (\ref{OutageProb}) are shown in Fig.~\ref{FigCDMA}. It shows that the Gaussian approximation of random interference is accurate and the outage probability increases sharply from $0$ to $1$ when $K \ge 5$. This implies that when there are more than $5$ active users per cell, the Wyner model characterizes the outage probability in the CDMA system very accurately.

\subsection{Average Throughput} \label{CapacitySCP}

In this subsection, to make a fair comparison with the available results obtained from the Wyner model, optimal multiuser detection (minimum mean square error detector plus successive interference cancelation) and perfect synchronization are assumed for the CDMA uplink with $G=K$, in which case the intracell interference is suppressed in effective, but the intercell interference still exists and is treated as Gaussian noise. Therefore, the average throughput in intracell TDMA and CDMA under the Wyner model are the same and given by \cite{Shamai97}
\begin{align}
\bar{R}_{\text{TDMA}}^{\star}= \bar{R}_{\text{CDMA}}^{\star} = \frac{1}{2K} \log \left( 1+ \frac{1}{2 \alpha^2} \right). \label{Wyner_Throughput}
\end{align}
After considering random user locations, the average throughput in intracell TDMA and CDMA become
\begin{align}
R_{\text{TDMA}} ^{\star} &= \frac{1}{2K} \mathbb{E} \left[ \log \left( 1+ \frac{1}{U_{n-1}^{(i)} +V_{n+1}^{(j)}} \right) \right], \nonumber \\
R_{\text{CDMA}}^{\star} &=  \frac{1}{2K} \mathbb {E} \left[ \log \left( 1+ \frac{1} { \frac{1}{K} \sum_{k=1}^{K} (U_{n-1}^{(k)}   +  V_{n+1}^{(k)}) } \right) \right].
\end{align}
The fact that TDMA outperforms CDMA and that random user locations improve average throughput is established by the following theorem.
\begin{theorem} \label{TheoremCDMACapcitySCP}
TDMA is advantageous versus CDMA in terms of average throughput and random user locations improve throughput under SCP, i.e.,
\begin{align}
R_{\text{TDMA}}^{\star}  \ge R_{\text{CDMA}}^{\star} \ge \bar{R}_{\text{TDMA}}^{\star} = \bar{R}_{\text{CDMA}}^{\star}, \nonumber
\end{align}
where the parameter $\alpha$ in the Wyner model is determined by (\ref{EquationDefAlphaUplinik}). The inequalities become equalities when the $\{ U_{n}^{(k)} \}$ and $\{ V_n^{(k)} \}$ are deterministic, i.e., the user locations are fixed and their distances from home BSs are the same.
\end{theorem}
\begin{proof}
See Appendix \ref{ProofTheorem1}.
\end{proof}

Simulation results in Fig.~\ref{ErgodicThroughput} shows that the Wyner model is quite accurate to characterize the average throughput in CDMA when $K > 20$ but inaccurate for TDMA. Also, it is shown that over all range of $K$ and $\beta$, intracell TDMA outperforms CDMA and random user locations increase the throughput substantially, e.g., when $\beta=4$, the throughput is increased by $100\%$. The Wyner model fails to capture either fact.

\section{Uplink with MCP} \label{UplinkMCP}

In this section, the average throughput in intracell TDMA and CDMA with MCP is investigated. The average throughput in intracell TDMA can be derived as
\begin{align}
C_{\text{TDMA}}^{\star}  = \lim_{N \to \infty} \frac{1}{KN} \mathbb{E} [I( \{ Y_n \}; \{ X_n \} )| \{L_n \} ]
            = \lim_{N \to \infty} \frac{1}{KN} \mathbb{E} \left[ \log ( \det ( \mathbf{ \Lambda_N} ) ) \right] ,
\end{align}
\noindent where $N$ is the number of cells and the expectation is taken over $N$ i.i.d. random user locations $\{L_n \}$. The $N \times N$ matrix $\sigma^2 \mathbf{\Lambda_N}$ is the covariance matrix of the conditioned Gaussian output vector $\{ Y_n \}$ and $\mathbf{\Lambda_N}$ is
\begin{align} \label{CovTDMA}
[ \mathbf{\Lambda_N} ]_{m,n} = \left \{
\begin{array}{ll}
 1+ KS ( 1 + U_{n-1}^{(i)} + V_{n+1}^{(j)}) &  (n,n) \\
 KS ( (U_{n}^{(k)})^{\frac{1}{2}}+ (V_{n+1}^{(j)})^{\frac{1}{2}} )  &  (n,n+1) \\
 KS ( (U_{n-1}^{(i)})^{\frac{1}{2}}+ (V_{n}^{(k)})^{\frac{1}{2}} )  & (n,n-1) \\
 KS ( (U_{n+1}^{(j)})^{\frac{1}{2}} (V_{n+1}^{(j)})^{\frac{1}{2}} ) & (n,n+2) \\
 KS ( (U_{n-1}^{(i)})^{\frac{1}{2}} (V_{n-1}^{(i)})^{\frac{1}{2}} ) & (n,n-2) \\
 0                     & \text{otherwise} \\
\end{array} \right. ,
\end{align}
where $K S= \frac{KP}{\sigma^2}$ is the SNR of each user in intracell TDMA.

Similarly, the average throughput in CDMA is given by
\begin{align}
C_{\text{CDMA}}^{\star} & =
\lim_{N \to \infty} \frac{1}{KN} \mathbb{E} [ I( \{ Y_n \}; \{ X_n^{(k)} \}  | \{ L_n^{(k)} \} )]
            = \lim_{N \to \infty} \frac{1}{KN} \mathbb{E} \left[ \log ( \det ( \mathbf{\Lambda_N}) )\right] , \nonumber
\end{align}
where the expectation is over $NK$ i.i.d. random locations of users $\{L_{n}^{(k)}\}$. The $N \times N$ matrix $\sigma^2 \mathbf{\Lambda_N}$ is the covariance matrix of the conditioned Gaussian output vector $\{ Y_n \}$ and $\mathbf{\Lambda_N}$ is
\begin{align}
[ \mathbf{\Lambda_N} ]_{m,n} = \left \{
\begin{array}{ll}
 1+ S \sum_{k=1}^K ( 1 + U_{n-1}^{(k)} + V_{n+1}^{(k)} ) &  (n,n) \\
 S \sum_{k=1}^K ( (U_{n}^{(k)})^{\frac{1}{2}}+ (V_{n+1}^{(k)})^{\frac{1}{2}} )  &  (n,n+1) \\
 S \sum_{k=1}^K ( (U_{n-1}^{(k)})^{\frac{1}{2}}+ (V_{n}^{(k)})^{\frac{1}{2}} )  & (n,n-1) \\
 S \sum_{k=1}^K ( (U_{n+1}^{(k)})^{\frac{1}{2}} (V_{n+1}^{(k)})^{\frac{1}{2}} ) & (n,n+2) \\
 S \sum_{k=1}^K ( (U_{n-1}^{(k)})^{\frac{1}{2}} (V_{n-1}^{(k)})^{\frac{1}{2}} ) & (n,n-2) \\
 0                     & \text{otherwise} \\
\end{array} \right. ,
\end{align}
\noindent where $S= \frac{P}{\sigma^2}$ is the SNR of each user. As expected, by setting $K=1$, the covariance matrix is the same in both systems.

In the following theorem, we prove that CDMA is advantageous to intracell TDMA.

\begin{theorem} \label{TheoremUplinkCapacity}
The average throughput of the CDMA system is larger than that of TDMA, i.e.,
$
C_{\text{CDMA}}^{\star} \ge C_{\text{TDMA}}^{\star} . \nonumber
$
Equality is achieved when the user locations are fixed and their distances from the home BSs are the same.
\end{theorem}
\begin{proof}
See Appendix \ref{ProofTheorem2}.
\end{proof}
\begin{remark}
In \cite{Wyner94}, it was shown that intracell TDMA is optimal and achieves the capacity under the Wyner model. However, when taking random user locations into consideration, intracell TDMA becomes suboptimal and the Wyner model leads to an inaccurate conclusion.
\end{remark}

\section{Downlink with SCP} \label{DownlinkSCP}

In this section, we will study the accuracy of the Wyner model in the downlink with SCP by analyzing the outage probability and average throughput. Intracell TDMA and CDMA under two different power allocation scenarios, namely perfect channel inversion (PCI) and equal transmit power per user (ETP), are considered.

\subsection{Perfect Channel Inversion (PCI) }
With SCP, the transmit signal at BS $n$ can be rewritten as a weighted summation of signals intended for its home users, that is, $X_n = \sum_{k=1}^K \sqrt {P_{n}^{(k)} } X_{n}^{(k)}$, where $P_{n}^{(k)}$ and $ X_{n}^{(k)}$ are respectively the transmit power and  signal for user $k$ in cell $n$ with $\mathbb{E} [ |X_{n}^{(k)}|^2] =1$. Since perfect channel inversion is assumed in this subsection, i.e., the received power of information signal of each user is a constant, equation (\ref{EquationReceivedPower}) holds.

Consider intracell TDMA. For the same reason as the uplink, Gaussian noise $Z_{n}^{(k)}$ is ignored in SCP. It follows that SIR of user $k$ in cell $n$ is a random variable given by
\begin{align}
\text{SIR}_n^{k} = \frac{ (r_{n,n}^{(k)} )^{-\beta} P_{n}^{(k) } } { ( r_{n,n-1}^{(k)} )^{-\beta} P_{n-1}^{(i) }+ (r_{n,n+1}^{(k)})^{-\beta} P_{n+1}^{(j)} }. \nonumber
\end{align}
Substituting (\ref{EquationReceivedPower}), the SIR can be simplified to
\begin{align}
\text{SIR}_n^{k} = \left[ A_{n-1}^{(i)}  + B_{n+1}^{(j)} \right]^{-1}, \nonumber
\end{align}
where $A_{n-1}^{(i)} = \left( \frac{ r_{n-1,n-1}^{(i)} }{ r_{n,n-1}^{(k)} } \right)^{\beta} $ and $B_{n+1}^{(j)} = \left( \frac{ r_{n+1,n+1}^{(j)} }{ r_{n,n+1}^{(k)} } \right)^{\beta}$.
Then, the outage probability $q$ is given by
\begin{align}
q  =  \mathbb{P}[\text{SIR}_n^{k}<\theta] = \mathbb{P} \left[ A_{n-1}^{(i)}  + B_{n+1}^{(j)}> \frac{1}{\theta} \right]. \nonumber
\end{align}
\noindent With the similar derivations as the uplink case, it can be lower bounded as
\begin{align}
 q \ge 1- \mathbb{P} \left[ A_{n-1}^{(i)}  \le \frac{1}{\theta}  \right]  \cdot \mathbb{P}  \left[ B_{n+1}^{(j)}\le \frac{1}{\theta} \right]. \label{equation:TDMA:downlinkoutage}
\end{align}
Since $L_{n}^{(i)} \sim  \mathbb{U}\:[-R, R]$, it can be easily derived as
\begin{align}
\mathbb{P} \left[ A_{n-1}^{(i)} \le \frac{1}{\theta}  \right]
 = \min \left \{   \left(1/ \theta \right)^{\frac{1}{\beta}} \left( r_{n,n-1}^{(k)} / {R} \right), 1  \right \}, \nonumber \\
\mathbb{P} \left[ B_{n+1}^{(j)} \le \frac{1}{\theta}  \right]
 = \min \left \{ \left(1 /\theta \right)^{\frac{1}{\beta}} \left( r_{n,n+1}^{(k)}/ {R} \right) , 1 \right \}, \nonumber
\end{align}
and the outage probability $q$ is lower bounded by (\ref{equation:TDMA:downlinkoutage}), which is tight as shown in Fig.~\ref{FigTDMANoPowerOutage}.

Similar to the uplink case, in the 1-D Wyner model for the downlink, $ \text{SIR} \equiv \frac{1}{2 \alpha^2}$ and there is no notion of outage. However, with random user locations, SIR becomes random and the outage probability changes with the SIR threshold as depicted in Fig.~\ref{FigTDMANoPowerOutage}. Also, it shows that in contrast to the uplink case, the outage probability of a particular user in the downlink also depends on its location. The Wyner model fails to capture these key facts.

Next, consider synchronous CDMA. The SIR for a particular user can be derived as
\begin{align}
\text{SIR}_n^{k}  &= \frac{ G (r_{n,n}^{(k)} )^{-\beta} P_{n}^{(k) } } { ( r_{n,n-1}^{(k)} )^{-\beta} \sum_{i=1}^K P_{n-1}^{(i) }+ (r_{n,n+1}^{(k)})^{-\beta} \sum_{j=1}^K P_{n+1}^{(j)} }  = \frac{ G } { \sum_{i=1}^K A_{n-1}^{(i)} + \sum_{j=1}^K B_{n+1}^{(j)} }, \nonumber
\end{align}
where $G$ is the processing gain in CDMA. When $K \gg 1 $, since $ r_{n-1,n-1}^{(i)}$ and $r_{n+1,n+1}^{(j)}$ are i.i.d., by the central limit theorem, $ \frac{1}{\sqrt{K} } \left( \sum_{i=1}^K A_{n-1}^{(i)} + \sum_{j=1}^K B_{n+1}^{(j)} \right) $ can be approximated by a Gaussian random variable with mean $\mu_r$ and variance $ \sigma_r^2$, which can be easily derived as
\begin{align}
\mu_r &=  \frac{\sqrt{K} }{(\beta+1)}  \left[ \left( R / r_{n,n-1}^{(k)} \right) ^\beta +  \left( R/ r_{n,n+1}^{(k)}  \right) ^\beta \right] , \nonumber \\
\sigma_r^2 & =   \frac{\beta^2 }{(\beta+1)^2(2 \beta +1) }  \left[ \left( R /r_{n,n-1}^{(k)}  \right) ^{2\beta} +  \left( R /r_{n,n+1}^{(k)}  \right) ^{2\beta} \right] . \nonumber
\end{align}
\noindent Therefore, the outage probability can be derived as
\begin{align}
q = \mathbb{P} \left[ \text{SIR}_n^{k} <\theta \right] \approx \mathcal{Q} \left(\frac{ G/{\theta} - \sqrt{K} \mu_r}{ \sqrt{K} \sigma_r} \right). \label{equation:CDMA:downlink:outage}
\end{align}

Define the average downlink intercell interference intensity with random user locations $\bar{\alpha}_{\text{dl}} $ as
\begin{align} \label{EquationDefAlphaCI}
\bar{\alpha}^{2}_{\text{dl}}= \mathbb{E} [A_{n-1}^{(i)} ] = \mathbb{E}[B_{n+1}^{(j)} ] = \frac{1} {(\beta+1) \gamma },
\end{align}
where $\gamma^{-1}= \mathbb{E} \left[ \left( R / r_{n,n \pm 1}^{(k)} \right) ^\beta \right]  = \frac{1- 3^{1-\beta}}{ 2 (\beta-1) }$.  If $K \to \infty$, due to the law of large numbers:
\begin{align}
\lim_{K \to \infty} \frac{1}{K}  \sum_{i=1}^{K}  A_{n-1}^{(i)} = \lim_{K \to \infty} \frac{1}{K}  \sum_{j=1}^{K} B_{n+1}^{(j)} = \gamma \; \bar{\alpha}^{2}_{\text{dl}}  \left(  R/  r_{n,n-1}^{(k)} \right)^\beta. \nonumber
\end{align}
Then, SIR becomes deterministic, which is given by
\begin{align}
\text{SIR}_n=\frac{G}{ \gamma \; \bar{\alpha}^{2}_{\text{dl}}  K \left( \left( R/ r_{n,n-1}^{(k)}  \right) ^\beta +  \left( R/ r_{n,n+1}^{(k)}  \right) ^\beta \right) } .
\end{align}
Thus, the outage probability is either $0$ or $1$ depending on both the SIR threshold and user locations; while in the Wyner model, it only depends on the SIR threshold.

Simulation and numerical results according to (\ref{equation:CDMA:downlink:outage}) are shown in Fig.~\ref{FigCDMADownlinkOutage}. It shows that the Gaussian approximation of random interference for a given user location is accurate and the outage probability increases sharply from $0$ to $1$ when $K \ge 50$. Compared to the results in the uplink, the new message here is that in the downlink the outage probability still varies over user locations and the Wyner model ignores this key fact.

Turn to throughput next. The average throughput in intracell TDMA and synchronous CDMA under the Wyner model are the same and simply given by (\ref{Wyner_Throughput}). With random user locations, the average throughput in intracell TDMA and synchronous CDMA with spreading gain $G=K$ becomes
\begin{align}
\bar{R}_{\text{TDMA}} &= \frac{1}{2K} \mathbb{E} \left[ \log \left(1+ \frac{1}{A_{n-1}^{(i)} + B_{n+1}^{(j)} } \right)  \right],  \\
\bar{R}_{\text{CDMA}} &= \frac{1}{2K} \mathbb{E} \left[ \log \left(1+ \frac{K}{ \sum_{i=1}^K A_{n-1}^{(i)} + \sum_{j=1}^K B_{n+1}^{(j)} } \right) \right].
\end{align}
The fact that by appropriately tuning parameter $\alpha$, the Wyner model is able to provide a lower bound to the average throughput with random user locations is established by the following theorem. The proof is similar to that of Theorem \ref{TheoremCDMACapcitySCP} and omitted for brevity.
\begin{theorem}
$
\bar{R}_{\text{TDMA}} \ge \bar{R}_{\text{CDMA}} \ge R^{\star}_{\text{TDMA}}=R^{\star}_{\text{CDMA}},
$
where $R^{\star}_{\text{TDMA}}$ and $R^{\star}_{\text{CDMA}}$ are defined in (\ref{Wyner_Throughput}), but the parameter $\alpha$ in the Wyner model is now chosen by (\ref{EquationDefAlphaCI}).
\end{theorem}

\subsection{Equal Transmit Power Per User (ETP) }
In this subsection, equal transmit power per user is assumed, where BSs are always transmitting at the fixed maximum power $P$. Then, it follows that the SIR of user $k$ in cell $n$ in intracell TDMA and synchronous CDMA with $G=K$ are the same and given by
\begin{align}
\text{SIR}_{\text{TDMA}}= \text{SIR}_{\text{CDMA}}= \frac{ (a_n^{(k)} )^2 P }{ ( b_n^{(k)} )^2 P+ (c_n^{(k)})^2 P }= \frac{1}{U_{n}^{(k)} + V_{n}^{(k)}}, \nonumber
\end{align}
where $U_{n}^{(k)}$ and $V_{n}^{(k)}$ are the same as defined in the uplink.

In contrast to perfect channel inversion, here the $\text{SIR}_n^{k}$ only depends on the location of user $k$ in cell $n$ and is independent of the user locations in neighboring cells. Therefore, strictly speaking, there is no outage here from a particular user's perspective. Thus, we will focus on the average throughput.

With random user locations, the average throughput in intracell TDMA and synchronous CDMA with $G=K$ are the same and given by
\begin{align}
\bar{R}_{\text{TDMA}} =\bar{R}_{\text{CDMA}} =  \frac{1}{2K} \mathbb{E} \left[ \log \left(1+ \frac{1}{U_{n}^{(k)} + V_{n}^{(k)} } \right) \right], \nonumber
\end{align}

Similar to perfect channel inversion, the Wyner model can provide a lower bound to the average throughput by tuning the parameter $\alpha$, which is established by the following theorem.
\begin{theorem}
$\bar{R}_{\text{TDMA}} = \bar{R}_{\text{CDMA}} \ge R^{\star}_{\text{TDMA}}=R^{\star}_{\text{CDMA}}$, where $G=K$ in CDMA and the parameter $\alpha$ in the Wyner model is given by (\ref{EquationDefAlphaUplinik}).
\end{theorem}
Simulation results of the average throughput are shown in Fig.~\ref{FigAverageThroughput}. It shows that the Wyner model can reasonably characterize the average throughput in CDMA with perfect channel inversion by appropriately tuning the parameter $\alpha$ according to (\ref{EquationDefAlphaCI}). However, it is inaccurate in other cases, especially for equal transmit power per user. The underlying reason is two-fold: (i) perfect channel inversion partially counteracts the effect of random user distributions over a cell; (ii) in CDMA with perfect channel inversion, the intercell interference does not depend on the specific user locations at the two neighboring cells due to the interference averaging effect as shown in Fig.~\ref{FigCDMADownlinkOutage}.

\section{Downlink with MCP} \label{DownlinkMCP}
In this section, the sum throughput with MCP in CDMA and intracell TDMA is investigated.
Rewrite (\ref{EquationWynerModelDownlink}) in the compact matrix form as
\begin{align}
{\bf Y} = {\bf H}^{T} {\bf X} + {\bf Z},  \label{EquationWynerModelDownlink_Matrix}
\end{align}
where ${\bf H}^{T}$ is the $NK \times N $ location-dependent channel matrix.

The downlink per-cell power constraint is given by
\begin{align}
( \mathbb{E} [{\bf X} {\bf X}^T] )_{ii} \le P, \qquad \text{for } i=1, \ldots, N .\label{EquationDownlinkPower}
\end{align}

The key tool used in the analysis of sum throughput is the minimax uplink-downlink duality theorem established in \cite{Yu06}, restated here for clarity.

\begin{theorem} \label{TheoremDuality}
(Minimax uplink-downlink duality \cite{Yu06}) For a given channel matrix $ \mathbf {H}$, the sum capacity of the downlink channel (\ref{EquationWynerModelDownlink_Matrix}) under per-cell power constraint (\ref{EquationDownlinkPower}) is the same as the sum capacity of the dual uplink channel affected by a diagonal ``uncertain'' noise under the sum power constraint :
\begin{align}
C^{\text{sum}} ( \mathbf{H}, S) = \min_{\mathbf{A}} \max_{\mathbf{D}} \log \left( \frac{\det ( \mathbf{H} \mathbf{D} \mathbf{H}^T+ \mathbf{A}) } {\det ( \mathbf{A}) } \right), \nonumber
\end{align}
where $S= \frac{P}{\sigma^2}$, $\mathbf{A}$ and $\mathbf{D}$ are $N$-dim and $NK$-dim nonnegative diagonal matrices such that $ \text{Tr} (\mathbf{A}) \le 1/S $ and $ \text{Tr} (\mathbf{D})  \le 1 $.
\end{theorem}

Consider CDMA first. Applying Theorem \ref{TheoremDuality}, the sum throughput with random user locations in CDMA is derived in the following corollary.
\begin{corollary}
The downlink per-cell sum throughput with random user locations in CDMA is
\begin{align}
C^{\text{sum}}_{\text{CDMA}} =\lim_{N \to \infty} \frac{1}{N } \mathbb{E}   \left[ \min_{ \mathbf{A} } \max_{\mathbf{D} } \log \left( \frac{\det ( \mathbf{H} \mathbf{D}  \mathbf{H}^T+ \mathbf{A} ) } {\det ( \mathbf{A}) } \right) \right], \nonumber
\end{align}
where the expectation is taken over all possible $NK \times N$ channel matrices $ \mathbf{H}^T$.
\end{corollary}

With random user locations, the channel matrices are random and thus asymmetric. Therefore, the above optimization problem has no simple and explicit solution. Thus, equal power allocation is assumed, i.e., $ \mathbf{D} = \frac{1}{KN} \mathbf{I} $, which gives a lower bound to the optimal solution:
\begin{align} \label{EquationCapacityRandom}
C^{l}_ {\text{CDMA}} =  \lim_{N \to \infty} \mathbb{E}_{\mathbf{H}} \left[ \frac{1}{N } \min_{\mathbf{A}}  \log \left( \frac{\det \left( 1/(KN) \mathbf{H} \mathbf{H}^T+ \mathbf{A} \right) } {\det ( \mathbf{A}) } \right) \right].
\end{align}
For arbitrary $K$, it is still hard to solve the above minimization problem, but if $K \to \infty$, the following lemma holds.
\begin{lemma} \label{lemmaCapacityRandomConvergence}
When $K \to \infty$, $\frac{1}{K}  \mathbf{H} \mathbf{H}^T$ converges to a deterministic and symmetric matrix, that is
\begin{align}  \label{EquationLambdaLimit}
\frac{1}{K}  \mathbf{H} \mathbf{H}^T  \overset{K \to \infty}{\longrightarrow}  \mathbf{\bar \Lambda},
\end{align}
where $\mathbf{\bar \Lambda}$ is defined in (\ref{DefBarLambda}).
\end{lemma}
\begin{proof}
See Appendix \ref{ProofLemma1}.
\end{proof}

Lemma \ref{lemmaCapacityRandomConvergence} enables us to derive the main theorem in this section.
\begin{theorem}
When $K \to \infty$, the lower bound to the per-cell sum throughput in CDMA converges to
\begin{align}
C^{l}_{\text{CDMA}} \overset{K \to \infty}{\longrightarrow} \lim_{N \to \infty} \frac{1}{N } \log \det( \mathbf{I} +  S \mathbf{\bar{\Lambda} } )  . \label{EquationCapacityLower}
\end{align}
\end{theorem}
\begin{proof}
By lemma \ref{lemmaCapacityRandomConvergence}, $\frac{1}{K} \mathbf{H} \mathbf{H}^T$ converges to a deterministic and symmetric matrix. Therefore, the minimum in (\ref{EquationCapacityRandom}) is achieved by substituting $\mathbf{A}= \frac{1}{NS} \mathbf{I} $. Then, the theorem follows by simple calculations.
\end{proof}

Let us compare the lower bound (\ref{EquationCapacityLower}) with the per cell sum throughput given by the Wyner model\cite{Shamai07}, which is
\begin{align} \label{EquationDownlinkCapacityWyner}
C^{\text{W}} = \lim_{N \to \infty} \frac{1}{N} \log \det   ( \mathbf{I} +  S \mathbf{\Lambda_0}) ,
\end{align}
where $ \mathbf{\Lambda_0}$ is defined as
\begin{align}
[  \mathbf{\Lambda_0} ]_{m,n} = \left \{
\begin{array}{ll} \nonumber
 (1 + 2 \alpha^2)    &  (n,n) \\
 2 \alpha   &  (n,n \pm 1) \\
 \alpha^2   & (n,n \pm 2) \\
 0                     & \text{otherwise} \\
\end{array} \right. .
\end{align}

For fair comparison, the channel gains are normalized to make $ \mathbb{E}[a_{n,k}^2] =1 $. Also, pick the intercell interference intensity according to (\ref{EquationDefAlphaUplinik}).
If we ignore the dependencies among random variables $(a_{n}^{(k)})^2$, $U_{n}^{(k)}$, $V_{n}^{(k)}$ and also assume $ \text{Var} [(U_{n}^{(k)})^{1/2}] =0$, then
\begin{align}
\mathbf{\bar \Lambda}_{n,n}  =  1 + 2 \alpha^2,  \quad \mathbf{\bar \Lambda}_{n,n \pm 1}   =   2 \alpha, \quad
\mathbf{\bar \Lambda}_{n,n \pm 2}  =    \alpha^2 . \nonumber
\end{align}
It follows that $\mathbf{\bar \Lambda}=  \mathbf{\Lambda_0}$ and $ C^{l} = C^{\text{W}}$.

From the above analysis, it shows that if the dependencies among channel gains are ignored and the variance of intercell interference intensity is assumed to be $0$, the lower bound coincides with the throughput given by the Wyner model. Moreover, if equal power allocation is assumed, the throughput with random user locations in CDMA is the same as that given by the Wyner model. However, generally (i) the channel gains depend on the user locations and are highly correlated; (ii) the variance of intercell interference intensity is far from $0$; (iii) unequal power allocation is often used. Therefore, the Wyner model can only characterize the sum throughput in CDMA approximately.

Simulation results in Fig.~\ref{FigCapacity} show the per-cell sum throughput given by the Wyner model (\ref{EquationDownlinkCapacityWyner}) and the lower bound with random user location (\ref{EquationCapacityLower}) in CDMA respectively. Also, a benchmark line corresponding to the case where there is no intercell interference ($\alpha=0$) is drawn. It shows that the Wyner model can characterize the lower bound quite accurately by picking the appropriate parameter $\alpha$ such that $\alpha^2 = \mathbb{E} [ U_{n}^{(k)} ]= \mathbb{E} [ V_{n}^{(k)} ] \approx \frac{1}{(\beta+1)2^{\beta}}$.

In the following, let us move to the intracell TDMA. With equal power allocation, the lower bound to the sum throughput in intracell TDMA can be derived as
\begin{align}
C^{l}_{\text{TDMA}} = \lim_{N \to \infty} \mathbb{E}_{ \mathbf{ \tilde {H}}}  \left[ \frac{1}{N } \min_{\mathbf{A}}  \log \left( \frac{ \det \left( (1/N)  \mathbf{ \tilde{H}} \mathbf{ \tilde{H}^T} + \mathbf{ A}  \right)} { \det ( \mathbf{ A } )} \right) \right],
\end{align}
where the expectation is taken over all possible $N \times N $ channel matrix $\mathbf{ \tilde {H}^T} $. Similar to the uplink case, we can derive the following theorem.
\begin{theorem}
With equal power allocation, the per cell sum throughput in CDMA is larger than that of TDMA, i.e.,
$
C^{l}_{\text{CDMA}} \ge C^{l}_{\text{TDMA}}.
$
Equality is achieved when the user locations are fixed and their distances from the home BSs are the same.
\end{theorem}
\begin{proof}
The proof is similar to that of Theorem \ref{TheoremUplinkCapacity}. It follows by Jensen's inequality and the fact that $ \min_{\mathbf{A}}  \log \left( \frac{ \det \left( (1/N)  \mathbf{ \tilde{H}} \mathbf{ \tilde{H}^T} + \mathbf{ A } \right)} { \det ( \mathbf{ A } )} \right) $ is a concave function of  $\mathbf{ \tilde{H}} \mathbf{ \tilde{H}^T} $.
\end{proof}
\begin{remark}
This theorem implies that with equal power allocation, intracell TDMA which is originally capacity-achieving under the Wyner model becomes suboptimal after considering random user locations. Also, in contrast to CDMA, in intracell TDMA, even when $K \to \infty$, the random channel matrix $\mathbf{ \tilde{H}} \mathbf{ \tilde{H}^T}$ does not converge to a deterministic limit. Therefore, we cannot properly tune the parameter $\alpha$ to make the Wyner model characterize the lower bound to the sum throughput in intracell TDMA accurately, as we did in CDMA.
\end{remark}

\section{Extension to OFDMA} \label{OFDMA}
In this section, OFDMA is considered and we find that in terms of the accuracy of the Wyner model, OFDMA lies between TDMA and CDMA. For concreteness, we explain this fact just for the uplink, but the downlink is similar.

OFDMA systems allocate users distinct time-frequency slices ({\it resource block}) consisting of $N_f$ subcarriers in frequency and $N_t$ consecutive OFDM symbols in time. Therefore, there is no intracell interference within a cell. However, with universal frequency reuse, there is intercell interference. Note that in a resource block, the allocation of subcarriers changes over OFDM symbols (due to frequency hopping).

The codeword of each user is modulated onto its resource block. Thus, the interference seen over transmitting a codeword may be attributed to transmissions from a total of $M ( 1\leq M \leq K )$ users in the neighboring cell. Note that in intracell TDMA, $M=1$; while in CDMA, $M=K$.

From the above analysis, the SIR can be derived as
\begin{align}
\text{SIR}_n  =  \frac{ M P_n^{(k)} (r_{n,n}^{(k)})^{-\beta} } { \sum_{i=1}^M P_{n-1}^{(i)} (r_{n-1,n}^{(i)})^{-\beta}+ \sum_{j=1}^M P_{n+1}^{(j)} (r_{n+1,n}^{(j)})^{-\beta} } = \frac{M} {\sum_{k=1}^{M} (U_{n-1}^{(k)}  + V_{n+1}^{(k)}) } \nonumber.
\end{align}
Following the same derivation in the CDMA case, when $M \gg 1$, the outage probability can be derived as
\begin{align}
q  =\mathbb{P} [\text{SIR}_n <\theta] = \mathbb{P} \left[\sum_{k=1}^{M} (U_{n-1}^{(k)}  + V_{n+1}^{(k)}) >\frac{M}{\theta} \right] \approx \mathbb{Q} \left( \frac{ M / {\theta}- \sqrt{M} \mu} { \sqrt{M} \sigma_u} \right), \nonumber
\end{align}
where $\mu$ and $\sigma_u^2$ are defined in (\ref{EquationSigmaU}) with replacement of $K$ by $M$. Also, when $M \to \infty$, $\text{SIR} \to \frac{1}{2\alpha^{2} }$, which is exactly the same as that given by the Wyner model.

Therefore, if $M \gg 1$, OFDMA resembles the CDMA model and the Wyner model is accurate to characterize the SIR distribution. From the simulation results in Fig.~\ref{FigCDMA}, it can be deduced that for such a conclusion to hold, $M$ may need to be larger than $5$ in the uplink. In practice, the actual value of $M$ in OFDMA systems depends on the specific resource allocation, frequency hopping, and ARQ (automatic repeat request) scheme design. For instance, it was argued that $M \approx K$ in Flash-OFDM system by properly designing the resource blocks for each BS (not the same among all the BSs) and assuming perfect symbol synchronization among transmissions of neighboring BSs (see Section $4.4$ in \cite{Tse05}). However, for the LTE cellular standards \cite{Ghosh10}, $M$ is typically much closer to $1$ than to $K$, since the resource blocks are fixed among all the BSs and the protocols are designed to allow for orthogonality amongst users in a cell, as well as potentially among two neighboring cells. Therefore in most implementations of LTE, we would expect the conclusions of the TDMA model to hold quite accurately. In general though, OFDMA systems will lie somewhere between TDMA and CDMA with the specific system design determining which they are more like.

\section{Conclusion and Future Work} \label{conclusion}
In this paper, we studied the accuracy of the 1-D Wyner model as compared to a 1-D model with random user locations. The main results are summarized in Table \ref{tapcap}. The $2$-D Wyner model would presumably be even less accurate vs. a 2-D grid or more randomly spatial counterparts \cite{AndBac10}, since more randomness due to user locations and interference levels would exist.  We also have not considered random channel effects (fading or shadowing) or many other possible permutations of power control, multiple access, or multi-cell processing.  In general, our results suggest that the more averaging there is -- for example due to spread spectrum, diversity, and/or mean-based metrics -- the more accurate the Wyner model is.

\appendix{
\subsection{Approximation of $\mu$ and $\sigma_u^2$ in section \ref{Section_Uplink_SCP_Outage} } \label{ApproximationAlpha}
We have $\mathbb{E} [U_{n}^{(k)} ] = \mathbb{E} \left[\left(\frac{r_{n,n}^{(k)} } {r_{n,n+1}^{(k)} } \right)^{\beta} \right] = \mathbb{E} \left[ \left(\frac{ |L_n^{(k)}| } {d-L_n^{(k)} } \right)^{\beta} \right]. \nonumber$
Consider the following approximation:
\begin{align}
\mathbb{E} \left[ \left(\frac{ |L_n^{(k)}| } {d-L_n^{(k)} } \right)^{\beta} \right] & \approx \frac{ E [|L_n^{(k)}| ]^{\beta} } {d^{\beta} }. \nonumber
\end{align}
The reason behind this approximation is that the numerator accounts more than the denominator for the value. Then, by uniform distribution of $L_n^{(k)}$,
it follows that
\begin{align}
\frac{ E [|L_n^{(k)}| ]^{\beta} } {d^{\beta} } &= \frac{2}{d^{\beta+1}} \int_{0}^{\frac{d}{2}} r^{\beta} \mathrm{d} r = \frac{1}{(\beta+1) 2^{\beta}}. \nonumber
\end{align}

\subsection{Proof of Theorem $1$} \label{ProofTheorem1}
\begin{align}
R_{\text{CDMA}}^{\star} & =  \frac{1}{2K} \mathbb {E} \left[ \log \left( 1+ \frac{1}{ \frac{1}{K} \sum_{k=1}^{K} ( U_n^{(k)} + V_n^{(k)}) } \right) \right] \nonumber \\
    & \overset{(a)}{\le} \frac{1}{K} \sum_{k=1}^{K} \frac{1}{2K} \mathbb {E} \left[ \log \left( 1+ \frac{1}{ U_n^{(k)} + V_n^{(k)} } \right) \right]
     =  R_{\text{TDMA}}^{\star} ,\nonumber
\end{align}
\begin{align}
R_{\text{CDMA}}^{\star} & = \frac{1}{2K} \mathbb {E} \left[ \log \left( 1+ \frac{1}{ \frac{1}{K} \sum_{k=1}^{K} (U_n^{(k)} + V_n^{(k)} ) } \right) \right] \nonumber \\
     & \overset{(b)}{\ge}   \frac{1}{2K} \log \left( 1+ \frac{1} { \frac{1}{K} \sum_{k=1}^{K} \mathbb{E} \left[  U_n^{(k)} +  V_n^{(k)} \right] } \right) \nonumber \\
     & \overset{(c)}{=} \frac{1}{2K} \log \left( 1+ \frac{1}{2 \alpha^2} \right) = \bar{R}_{\text{TDMA}}^{\star}= \bar{R}_{\text{CDMA}}^{\star} . \nonumber
\end{align}
where $(a)$ and $(b)$ follows from Jensen's inequality and the fact that $\log( 1+ \rho x ^{-1})$ is a convex function of $x > 0$ (for any $\rho \ge 0$), and $(c)$ follows from the definition of $\alpha^2$.

\subsection{Proof of Theorem $2$} \label{ProofTheorem2}

Rewrite the covariance matrix $\mathbf{\Lambda_N}$ for CDMA as a summation of $K$ covariance matrices $\mathbf{\Lambda_N^k}$ as $\mathbf{\Lambda_N}= \frac{1}{K} \sum_{k=1}^K \mathbf{\Lambda_N^k} \label{LambdaN}$, where each $\mathbf{\Lambda_N^k}$ corresponds to a covariance matrix given by (\ref{CovTDMA}) in intra-cell TDMA. Then it follows that
\begin{align}
C_{\text{CDMA}}^{\star} & = \lim_{N \to \infty} \mathbb{E} \log ( \det ( \mathbf{\Lambda_N} ) )  = \lim_{N \to \infty}\mathbb{E} \log ( \det ( \frac{1}{K} \sum_{k=1}^K \mathbf{\Lambda_N^k} ) )  \nonumber \\
& \overset{(a)}{\ge} \frac{1}{K} \sum_{k=1}^K \lim_{N \to \infty} \mathbb{E} \log ( \det (\mathbf{\Lambda_N^k} ) ) =\frac{1}{K} \sum_{k=1}^K C_{\text{TDMA}}^{\star} = C_{\text{TDMA}}^{\star} , \nonumber
\end{align}
where $(a)$ follows from Jensen's inequality and the fact that $\log (\det \mathbf{\Lambda})$ is concave for the covariance matrix $\mathbf{\Lambda}$.

\subsection{Proof of Lemma $1$} \label{ProofLemma1}

The $(n,n)$th entry of $\frac{1}{K}  \mathbf{H} \mathbf{H}^T$ is
\begin{align}
\left[ \frac{1}{K}  \mathbf{H} \mathbf{H}^T \right]_{n,n} = \frac{1}{K} \sum_{k=1}^K \left(  (a_{n}^{(k)})^2 +  (b_{n}^{(k)})^2 + (c_{n}^{(k)})^2  \right) . \nonumber
\end{align}
Applying the law of large number, it follows that
\begin{align}
\left[ \frac{1}{K}  \mathbf{H} \mathbf{H}^T \right]_{n,n} \overset{K \to \infty}{\longrightarrow}   \mathbb{E} [(a_{n}^{(k)})^2] + \mathbb{E} [(b_{n}^{(k)})^2] + \mathbb{E} [(c_{n}^{(k)})^2] =\mathbb{E} \left[ (a_{n}^{(k)})^2 \left( 1 + U_{n}^{(k)}+ V_{n}^{(k)} \right) \right] . \nonumber
\end{align}
Finally, the lemma establishes by following the same derivations for other entries of $\frac{1}{K}  \mathbf{H} \mathbf{H}^T$ and $ \mathbf{\bar \Lambda}$ can be found as
\begin{align} \label{DefBarLambda}
[\mathbf{\bar \Lambda}_{m,n} ]
= \left \{
\begin{array}{ll}
  \mathbb{E} \left[ (a_{n}^{(k)})^2 \left( 1 + U_{n}^{(k)}+ V_{n}^{(k)} \right) \right]  &  (n,n) \\
   \mathbb{E}[(a_{n}^{(k)})^2 (U_{n}^{(k)})^{\frac{1}{2}} ] + \mathbb{E}[(a_{n}^{(k)})^2 (V_{n}^{(k)})^{\frac{1}{2}} ]  &  (n,n \pm 1) \\
  \mathbb{E}[( a_{n}^{(k)})^2 (U_{n}^{(k)})^{\frac{1}{2}} (V_{n}^{(k)})^{\frac{1}{2}} ]     & (n,n \pm 2) \\
 0                     & \text{otherwise} \\
\end{array} \right. .
\end{align}
}

\bibliography{BibModelingCelluarNetworks}

\begin{table}
\centering
\caption{The Accuracy of the Wyner model in cellular networks }
\label{tapcap}
\begin{tabular}{l|lll|lll}
\hline
                &                      &  Uplink  &          &              &Downlink       &   \\
\cline{2-7}
                &      Intracell TDMA  & CDMA  & OFDMA     &Intracell TDMA  &CDMA           &OFDMA     \\
\hline
OP with SCP and PCI     &      Low      & High  &  Medium           & Low           &Low            & Low     \\
\hline
AvgTh with SCP and PCI  &      Low       & High  & Medium           & Low           &Medium         &Medium  \\
\hline
AvgTh with SCP and ETP  & None                &None    &None      &Low         &Low            &Low \\
\hline
AvgTh or SumTh with MCP  &     \multicolumn{3} {l|} {Medium (intracell  TDMA becomes suboptimal)}   & \multicolumn{3}{l}   {Medium (capture lowerbound accurately in CDMA )} \\
\hline
\end{tabular}
\end{table}

\begin{figure}
\centering
\includegraphics[width=4.5in]{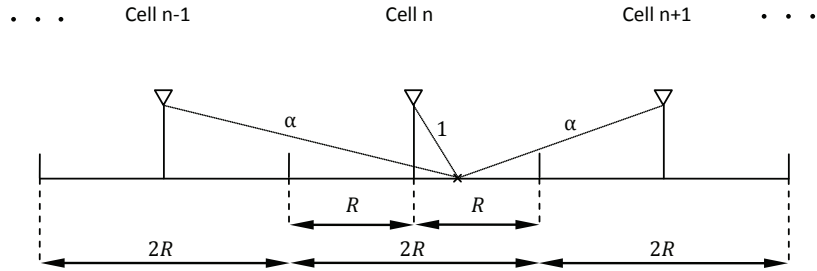}
\caption{The 1-D Wyner model.}
\label{FigWynerModel}
\end{figure}

\begin{figure}
\centering
\includegraphics[width=4.5in]{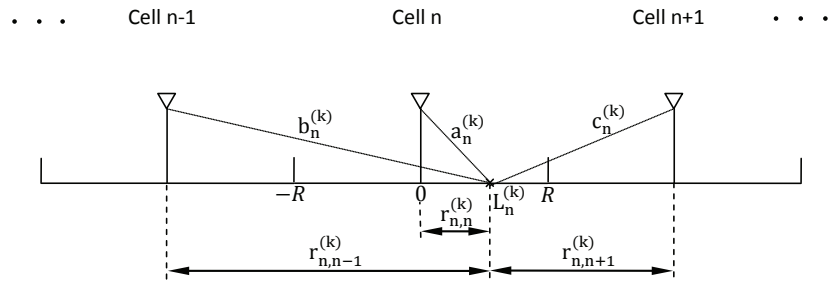}
\caption{The 1-D grid-based model used for comparison, which has random user locations.}
\label{FigWynerModelRandom}
\end{figure}

\begin{figure}
\centering
\includegraphics[width=4.5in]{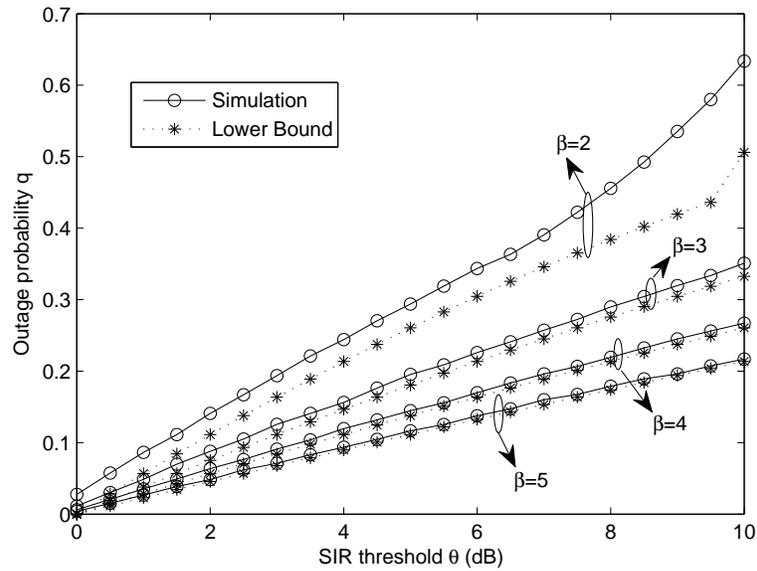}
\centering
\caption{Outage probability in intracell TDMA for the uplink with varying SIR threshold $\theta$.}
\label{FigTDMA}
\end{figure}

\begin{figure}
\centering
\includegraphics[width=5in]{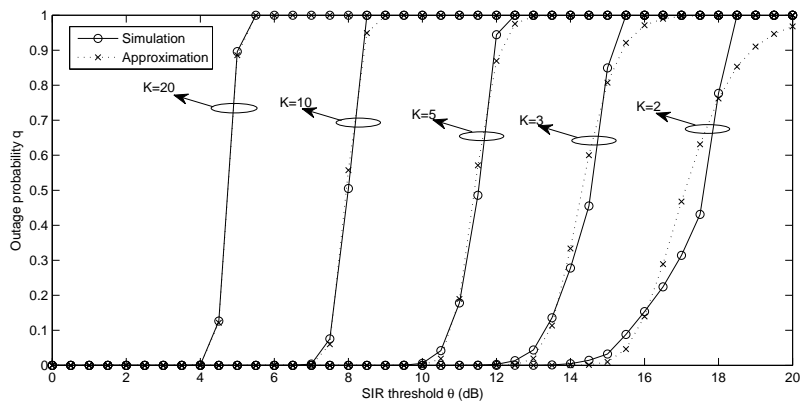}
\centering
\caption{Outage probability in CDMA for the uplink with varying number of users $K$ per cell, pathloss exponent $\beta=4$, $G=64$.}
\label{FigCDMA}
\end{figure}

\begin{figure}
\centering
\includegraphics[width=4.5in]{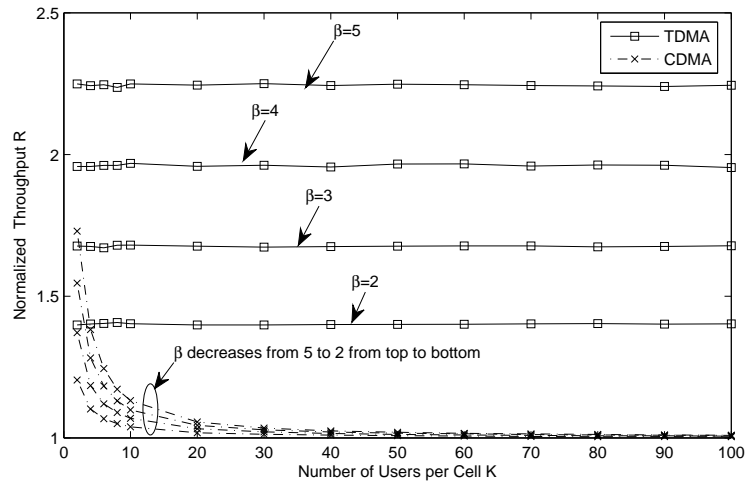}
\centering
\caption{The average throughput normalized by that given by the Wyner model in (\ref{Wyner_Throughput}) for the uplink with varying number of users $K$ per cell, $G=K$ in CDMA.}
\label{ErgodicThroughput}
\end{figure}

\begin{figure}
\centering
\includegraphics[width=4.5in]{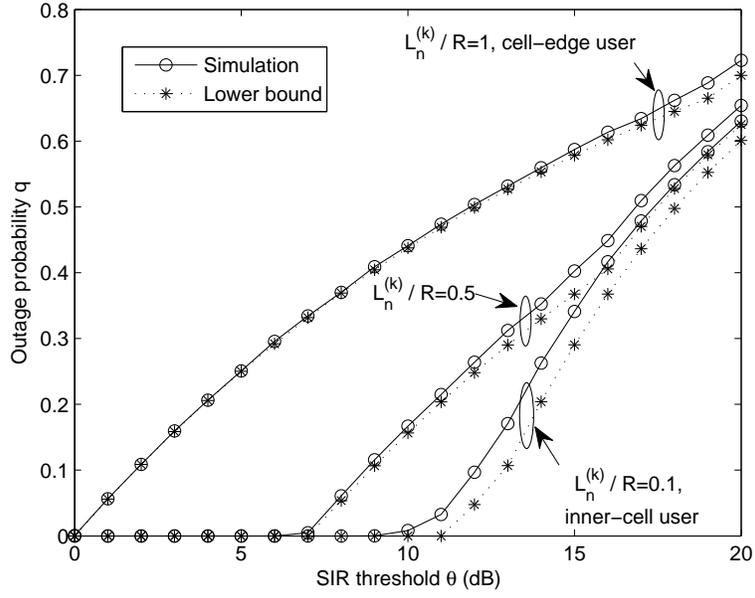}
\centering
\caption{Outage probability in intracell TDMA for the downlink with varying SIR threshold $\theta$ and user locations, $\beta=4$. }
\label{FigTDMANoPowerOutage}
\end{figure}

\begin{figure}
\centering
\includegraphics[width=4.5in]{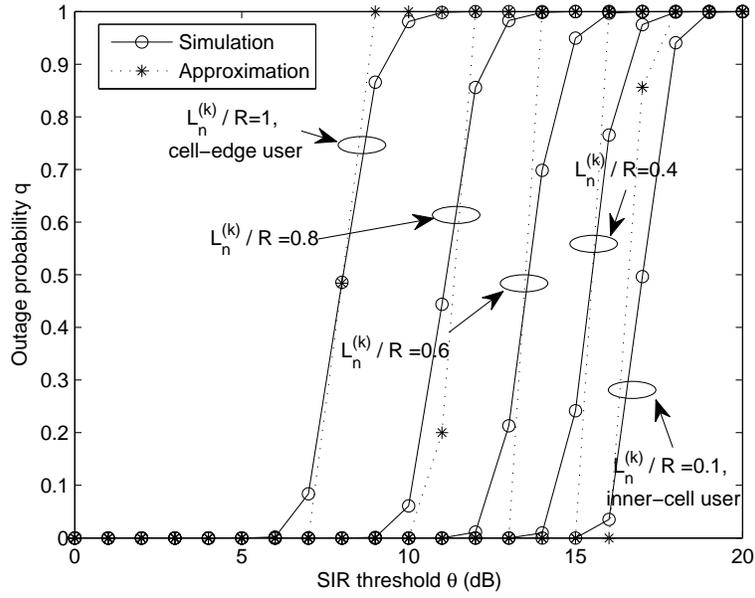}
\centering
\caption{Outage probability in CDMA for the downlink with varying SIR threshold $\theta$ and user locations, $K=50$, $\beta=4$, $G=64$.}
\label{FigCDMADownlinkOutage}
\end{figure}

\begin{figure}
\centering
\includegraphics[width=4.5in]{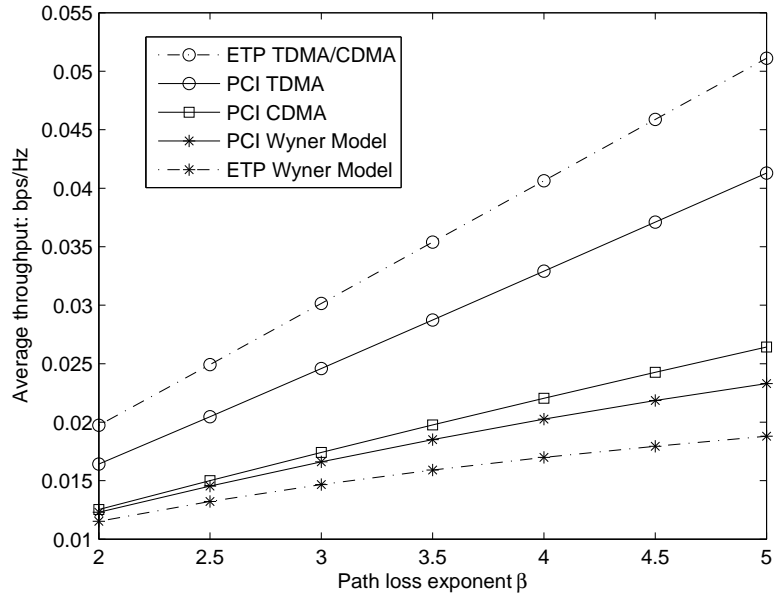}
\centering
\caption{Average throughput with varying pathloss exponent $\beta$, $G=K=100$ in CDMA.}
\label{FigAverageThroughput}
\end{figure}

\begin{figure}
\centering
\includegraphics[width=4.5in]{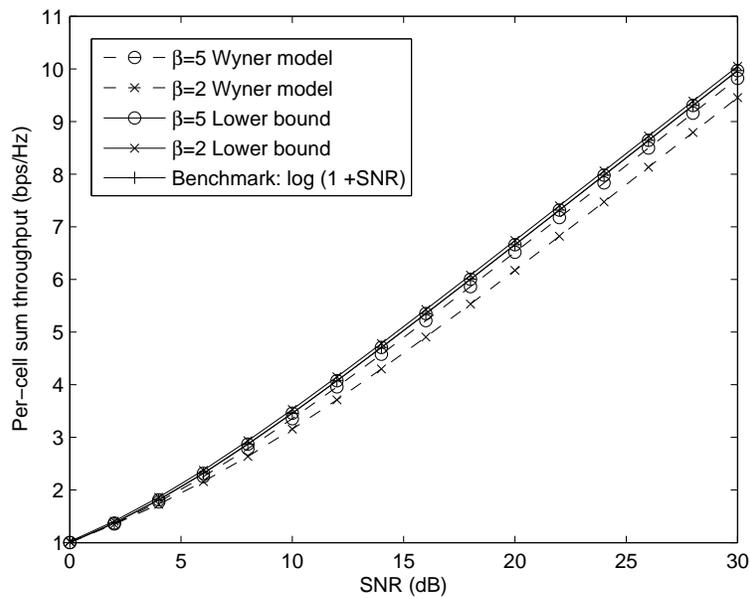}
\centering
\caption{Per-cell sum throughput for the downlink with varying SNR and pathloss exponent $\beta$.}
\label{FigCapacity}
\end{figure}

\end{document}